\documentclass[11pt]{article}  
\usepackage{amsfonts,amsmath,amssymb,amsthm,graphicx}
\usepackage{fullpage}
\usepackage{changepage}
\usepackage{enumerate}
\usepackage{algorithm}
\usepackage{algorithmic}
\usepackage{comment}
\usepackage{cleveref}
\usepackage{natbib}
\usepackage{algorithm}
\usepackage{algorithmic}
\newcommand{\cM}{{\mathcal M}}

\newcommand{\bM}{{\mathbb{M}}}
\newtheorem{theorem}{Theorem}
\newtheorem{lemma}{Lemma}
\newtheorem{informal}{Informal Theorem}
\newtheorem{property}{Property}

\begin{document}

\title{Revenue  Maximization  with  Imprecise  Distribution}  



\author{Yingkai Li
  \thanks{Department of Electrical Engineering and Computer Science, Northwestern University. Email: \texttt{yingkai.li@u.northwestern.edu}.}
  \and
  Pinyan Lu
  \thanks{Institute for Theoretical Computer Science, Shanghai University of Finance and Economics.
  Email: \texttt{lu.pinyan@mail.shufe.edu.cn}. 
  This author is supported by Innovation Program of Shanghai Municipal Education Commission.}
  \and 
  Haoran Ye
  \thanks{Zhiyuan College, Shanghai Jiao Tong University.
  Email: \texttt{yepro12@sjtu.edu.cn}.} }


\maketitle

\begin{abstract}  
We study the revenue maximization problem with an imprecisely estimated distribution of a single buyer or several independent and identically distributed buyers given that this estimation is not far away from the true distribution. We use the earth mover's distance to capture the estimation error between those two distributions in terms of both values and their probabilities, i.e., the error in value space given a quantile, and the error in quantile space given a value. 
We give explicit characterization of the optimal mechanisms for the single buyer setting. For the multi-buyer case, we provide an algorithm that finds an approximately optimal mechanism (FPTAS) among the family of second price mechanisms with a fixed reserve.
\end{abstract}


\section{Introduction}
One of the most important topics in mechanism design is the revenue maximization problem.
Optimal mechanisms have been well studied in both 
single-item setting, following the seminar work of \citet{myerson1981optimal}, 
and multi-item setting,
see \citet{cai2012algorithmic,cai2012optimal}. 
There is also a huge literature on approximately optimal simple mechanisms,
see \citet{roughgarden2018approximately} for a detailed discussion. 
All those papers assume that the seller knows the exact distribution of the buyers.
However, according to \citet{harsanyi1967games}, this assumption is hard to realize in practice.
To overcome this problem, lots of studies have focused on revenue maximization with weaker assumptions.
In \citet{cole2014sample},
they assume that the seller only has access to the samples of the buyers' distributions,
and the goal is to maximize the revenue using the distribution reconstructed from the samples.
The model in \citet{chen2017query} assumes that the seller has oracle access to the value and quantile information about the true distribution.
Similar to the goal of sample complexity, the goal here is to reconstruct the distribution with limited queries. 
The prior-independent mechanism design model considered in \citet{devanur2011prior} 
adopts an extreme assumption that the seller has no information about the distribution. 
They assume that the buyers' valuations are independent and identically distributed.
All those papers focus on the nature of the imprecision on the prior distribution, which will also be the main focus of our paper. 
\subsection{Our Problem and Results}
The seller is given an imprecisely estimated distribution of a single buyer or several i.i.d buyers, 
and the assumption is that it is
within $\epsilon$ earth mover's distance to the true distribution. The seller designs a mechanism based on the estimated distribution,
while its performance is evaluated with respect to the worst case distribution within $\epsilon$ earth mover's distance to the estimated distribution. We want to design the optimal mechanism in this worst case performance metric.  This will give a lower bound for any true distribution in that range and we call such mechanism \emph{the most robust mechanism}.

In section \ref{sec_single}, we solve the max-min problem when there is a single buyer.
The idea is to write this problem using a linear program, and transform it into its dual form.
By analyzing the properties of the dual, we can successfully characterize the optimal robust mechanism.
Here, we state the results for continuous distributions informally.
\begin{informal}
When there is a single buyer, the allocation $x$ for the optimal mechanism satisfies that $\exists a, b$: 
\begin{eqnarray*} x(v) =\begin{cases} 0& \forall v < a;\\
 (\ln \frac{v}{a}) / (\ln \frac{b}{a})& \forall a \leq v \leq b;\\
1& \forall b < v.
\end{cases}
\end{eqnarray*}
\end{informal}

The formal result are stated in Theorem \ref{thm_1buy} for discrete distributions. For continuous distribution, we can discretize the support and come up with such mechanism that is arbitrarily close to the optimal. 
 As we show in the above theorem, the optimal mechanism is a randomized mechanism.
Moreover, using this characterization, we can bound the gap between the optimal randomized mechanism
and the optimal deterministic mechanism for the max-min goal,
which is shown in Theorem \ref{thm_d_up} and \ref{thm_d_low}.
In section \ref{sec_multi},
we first characterize the worst case distribution for the second price mechanism.
\begin{informal}
When there are multiple buyers, the worst case distribution 
for the second price mechanism satisfies that, 
there exists $k,l$ such that $k < l$, 
the density between value $k$ and $l$ is $0$, 
and the distribution outside $[k, l]$ remains the same 
comparing to the distribution known by the seller.
\end{informal}
This result is formally stated in Theorem \ref{thm_sec}.
Using this characterization as a bridge,
we are able to further characterize the distribution with minimum revenue for second price mechanism with a fixed reserve.
Moreover, in Theorem \ref{thm_reserve}, for Lipschitz continuous distributions with bounded support,
we design a FPTAS algorithm for finding the \emph{optimal robust reserve} for the second price auction
when there are multiple buyers. 
\citet{ostrovsky2011reserve} show that 
by decreasing the theoretical optimal reserve price 
by a constant factor, 
the revenue increases in practice. 
They do not provide theoretical intuition for their result. 
Our result indicates that one possible reason is that 
the underlying true distribution is not identical to the 
distribution known by the seller, 
and the original theoretical optimal reserve price 
is not robust with respect to the imprecision of the distribution. 

\subsection{Previous Work on Robust Mechanism Design}
There are a number of works studying the robustness of mechanisms at different levels of prior imprecision
in both econ literature, 
e.g., \citet{Dirk2008Robust,Carroll2016Robustly,Carroll2017Robustness}, 
and computer science literature, 
e.g., \citet{cai2017learning, gravin2018separation}. 
To quantify the imprecision, some may define a distance metric for distributions, and then an imprecise distribution can be characterized as any distribution within a distance of $\epsilon$ from the true one. 
With those imprecise distributions, we can design a max-min game between the seller and the adversarial nature: 
the seller proposes a mechanism, and then the adversary chooses a distribution with minimum expected revenue generated by that mechanism. 
The performance of the optimal mechanism is studied under such circumstance.
For one dimensional density functions $f,f'$ and their corresponding cumulative distributions $F,F'$, we present the informal definition of some metrics used in previous works:
\begin{itemize}
\item{Total Variation distance (TV distance)}: $\sup_{S\subset[0,+\infty)}|f(S)-f'(S)|$
\item{Kolmogorov distance}: $\sup_{x\geq0}|F(x)-F'(x)|$
\item{Prokhorov distance}: $\inf\{\epsilon|\forall S\subset[0,\infty), f(S)\leq f'(S^\epsilon)+\epsilon\}$,
where $S^\epsilon=\{x\in[0,+\infty)|\min_{y\in S}|x-y|\leq\epsilon\}$.
\item{Earth Mover's distance (First order Wasserstein distance)}:
$\inf_{\pi(x,y)\in\Pi}\int_0^{+\infty}\int_0^{+\infty}\pi(x,y)|x-y|dxdy$
where $\Pi$ is the set of all the two dimensional distribution whose marginals are respectively $f$ and $f'$.
\end{itemize}
With these distance metrics, some researches have been conducted on robust mechanisms design. 
\citet{cai2017learning} studies the revenue guarantee of simple and robust mechanism in multi-item, multi-buyer settings. 
By sampling from the distribution $O(\frac{log(\frac{1}{\delta})}{\epsilon^2})$ times, 
they are able to learn the marginal distribution for each item that is within $\epsilon$ TV-distance, with probability at least $1-\delta$. 
In such setting, they also study the upper bound of approximation ratio for simple mechanisms, such as rationed sequential posted price mechanisms or anonymous sequential posted price mechanisms with entry fees.
The work of \citet{Dirk2008Robust} focuses on the Prokhorov distance. 
When there is a single buyer, a single item and a potential $\epsilon$ Prokhorov distance between the true prior and the known distribution, 
the paper examines the properties and the equilibrium of the maxmin game via an implicit function, 
and they characterize the seller's optimal mechanism as posting a deterministic price. 

Most of the distance metrics are chosen for the convenience of computation, while disadvantages of such choices do exist.
\begin{itemize}
\item\textbf{Sensitivity to Distribution:}
Some distance measures only focus on the errors in probability space, not the errors in value space.
For instance, consider two distributions which are $1$ and $1-\epsilon$ respectively with probability $1$, where $\epsilon>0$ is a small constant.
Those two distributions should be considered to be close to each other, with a small error in value space.
However, the TV distance and Kolmogorov distance between these two distributions both achieve their maximum $1$, indicating that a tiny deviation of distributions can bring about a giant variation distances. 
\item\textbf{Sensitivity to Structure:}
Some previous results for certain metrics benefit from the property that there is a distribution which is statistically dominated by all the distributions within certain distance. For Prokhorov distance and Kolmogorov distance, such distributions is to move the the probability of measure $\epsilon$ from the highest value to the lowest. 
For single item settings, the max-min mechanism becomes the optimal mechanism on the dominated distribution due to revenue monotonicity. 
This property does not hold in general, such as the robustness for the earth mover's distance, 
and hence their results cannot be generalized for more general settings. 
\end{itemize}
Therefore, proper distance measures should take the errors in value space into consideration, i.e. either there is a small error in estimating the value or the probability of a large error in estimating the value is small. 
This property is well characterized by the earth mover's distance. 
Moreover, there does not exist a single distribution that is stochastically dominated by all other distributions within $\epsilon$ earth mover's distance to the known distribution. 
Thus, we cannot exploit the revenue monotonicity of the single item auction. 
As a result, our mechanism is much more complex than the previous results, 
in the sense that our mechanism is randomized. 

\section{Preliminaries}
In this paper, we consider the problem that the seller tries to sell a single indivisible item to the buyers whose values are independently and identically distributed.
Let $M = \{1,\dots,m\}$ denote the set of buyers.
For the sake of simplicity, we start with discrete distributions,
and we will show how to generalize our results to continuous distributions.
For any buyer $j \in M$, his value $v_j$ takes from a discrete value set
$V = \{v_0, \dots, v_n\}$, where $v_0 < v_1<\cdots < v_n$.
We assume that the seller knows the set $V$ and a discrete distribution $f = (f_i)_{i\in \{0,...,n\}}$, where $f_i$ is the probability of value $v_i$.
When we consider the continuous distribution,
we will use the notation and let $f(v)$ denote the probability density at value $v$.
We denote $F$ as the corresponding cumulative probability function.
Note that this is not the true underlying distribution for the buyers.

We assume that the error of the distribution is small.
We characterize this error using earth mover's distance.
To change from $f$ to $f'$, we need to move a possibility measure at least $ \Big| \sum_{j < i}(f_j - f'_j) \Big|$ from $v_i$ to $v_{i-1}$ or from $v_{i-1}$ to $v_{i}$. 
Therefore, the earth mover's distance between discrete distributions $f$ and $f'$ is equivalently defined as
\begin{equation}\label{eq_emd}
EMD(f, f') = \sum_{i\in[n]} (v_i - v_{i-1}) \Big| \sum_{j < i}(f_j - f'_j) \Big|
\end{equation}
We assume that the true distribution for the buyers is within $\epsilon$ earth mover's distance of the known distribution $f$.
That is, the seller knows $f$ and knows that the true distribution $f' \in EMD(f, \epsilon)$,
where
\begin{equation}\label{eq_emd_eps}
EMD(f, \epsilon) = \{f' | EMD(f, f') \leq \epsilon\}
\end{equation}
Since both probability density function $f$ and cumulative probability function $F$ uniquely decides a distribution,
we also use notations $EMD(F, F')$, $EMD(F, \epsilon)$ for the same meaning.

In this paper, we use $\cM = (x, p)$ to denote a mechanism,
where $x = (x^j_i)_{i\in \{0,...,n\}, j\in M}$ is the allocation rule
and $p = (p^j_i)_{i\in \{0,...,n\}, j\in M}$ is the payment rule.
Here $x^j_i$ is the probability that buyer $j$ gets the item when he bids $v_i$,
and $p^j_i$ is the price that buyer $j$ pays when he bids $v_i$.
When there is only one buyer, we will omit the superscription for the allocation and the price without ambiguity.

The goal of the seller is to find a individual rational (IR) incentive compatible (IC) mechanism $\cM$
to maximize the minimum revenue among all those possible distributions.
That is, denoting $Rev(\cM, f)$ as the revenue of a IR-IC mechanism $\cM$ with buyer distribution $f$,
and denoting $\bM$ as the set of mechanisms we consider for the setting,
the goal of the seller is to find a mechanism $\cM^*$
such that
\begin{equation}\label{opt_1}
\cM^* = \arg\max_{\cM \in \bM}\min_{f'\in EMD(f, \epsilon)} Rev(\cM, f')
\end{equation}

In the rest of this paper, we will show how to achieve this goal when there is only one buyer
and the mechanism set $\bM$ is the set of all IR-IC mechanisms,
or when there are multiple buyers
and the mechanism set $\bM$ is the set of second price mechanisms with fixed reserves.
Moreover, for the single buyer case, we start with the assumption that the distribution $f$ known by the seller is regular,
i.e. the virtual value $\phi_i = v_i-\frac{1-F_i}{f_i}$ is non-decreasing \cite{myerson1981optimal}.
Note that this does not imply that the true distribution is regular.
In fact, this assumption can be removed and we formally discuss it in section \ref{sec_irregular}. 

\section{Single Buyer Case}\label{sec_single}
When there is a single buyer, the form of the mechanism is simply posting a (randomized) menu to the buyer,
and let the buyer choose his favorite entry.
In this case, we prove that the mechanism satisfying the max-min goal has a very simple form.
Formally, we have the following theorem.

\begin{theorem}\label{thm_1buy}
When there is a single buyer and $\bM$ is the set of all IR-IC mechanisms,
for any discrete distribution $f$ with support $\{v_0, v_1,\dots,v_n\}$,
there exists a polynomial time algorithm that finds mechanism
$\cM^* = \arg\max_{\cM \in \bM}\min_{f'\in EMD(f, \epsilon)} Rev(\cM, f')$.
Moreover, mechanism $\cM^* = (x,p)$ takes one of the following forms.
\vspace{-3pt}
\begin{enumerate}
\item For any $i\in \{0,\dots,n\}, x_i = 1, p_i = v_0$.
\item There exist $0 \leq a \leq b \leq n$ such that
$\lambda = 1 / (\sum_{j=a}^{b} \frac{v_j - v_{j-1}}{v_j})$, and
\begin{equation*}
x_i=\begin{cases}
0& 0 \leq i < a\\
\sum_{j=a}^i\frac{\lambda(v_j - v_{j-1})}{v_j} & a \leq i \leq b\\
1& b < i \leq n.
\end{cases}
\end{equation*}
\end{enumerate}
\end{theorem}

In order to solve problem (\ref{opt_1}), first we transform the problem into a simpler form.
According to \cite{zhiyi2016side}, the revenue of the seller is monotone with respect to statistical dominance.
Therefore, we only need to consider the distribution where the probability mass is moved from a high value to low value.
Let the probability transferred from value $v_i$ to $v_{i-1}$ be $t_i$. Then the new distribution $f'$ satisfies that $f'_i =f_i + t_{i+1}-t_i $ and the earth mover distance between $f$ and $f'$ can be written as 
\[EMD(f, f')=\sum_{i=1}^n t_i (v_i - v_{i-1}). \]
Let the allocation and payment rule of a mechanism $\cM$ be $(x, p)$. We can write (\ref{opt_1}) as the following explicit mathematical programming: 

\begin{equation}\label{opt_2}
\max_{x,p}\min_{t} \sum_{i=0}^np_if_i - \sum_{i=1}^nt_i(p_i-p_{i-1})
\end{equation}
\begin{eqnarray*}
v_ix_i-p_i\geq v_i x_j-p_j, && \forall 0 \leq i,j \leq n\\
v_ix_i-p_i\geq 0, && \forall 0 \leq i \leq n\\
p_i\geq0, 0 \leq x_i \leq 1, && \forall 0 \leq i \leq n\\
t_i+f_{i-1}\geq t_{i-1}, && \forall 2 \leq i \leq n\\
t_i\geq0, && \forall 1 \leq i \leq n\\
f_n \geq t_n, 
\sum_{i=1}^n t_i (v_i - v_{i-1}) \leq \epsilon. &&
\end{eqnarray*}

\noindent
Fix a feasible pair of variables $x, p$ and focus on the
minimization part of the optimization problem~(\ref{opt_2}).
Since the term $\sum_{i=1}^n p_i f_i$ is a constant with respect to $t$,
the minimization problem can be written equivalently into the following form.
\begin{equation}\label{opt_3}
\min_{t}-\sum_{i=1}^nt_i(p_i-p_{i-1})
\end{equation}
\vspace{-10pt}
\begin{eqnarray*}
t_{i+1} + f_{i} \geq t_{i} 
, && \forall 1 \leq i \leq n-1\\
t_i\geq 0, && \forall 1 \leq i \leq n\\
f_n \geq t_n, 
\sum_{i=1}^n t_i (v_i - v_{i-1}) \leq \epsilon
. &&
\end{eqnarray*}
We can write its dual form as follows:
\begin{equation}\label{opt_4}
\max_{\beta,\lambda}-\lambda\epsilon - \sum_{i=1}^{n}\beta_i f_{i}
\end{equation}
\vspace{-10pt}
\begin{eqnarray*}
\beta_i - \beta_{i-1} + \lambda (v_i - v_{i-1})
\geq p_i-p_{i-1}, && \forall 2 \leq i \leq n\\
\beta_1 + \lambda (v_i - v_{i-1})
\geq p_1 - p_0, &&\\
\lambda\geq0, 
\beta_i\geq0, && \forall 1 \leq i \leq n\\
\end{eqnarray*}
We define $z = \{z_i\}_{i\in \{0,\dots,n\}}$ such that $z_0 = p_0$ and $z_i = p_i-\beta_i, \forall 1\leq i \leq n$.
Also, we say a vector $v$ dominates $u$ if and only if for any $i$, $v_i\geq u_i$.
According to Myerson's lemma, in a DSIC mechanism, the allocation is non-decreasing.
We denote $P$ as the set of all possible payment rules of the IR-IC mechanisms,
and we denote $Z$ as the set of all possible $z$'s that are dominated by some $p\in P$.
Obviously, $P\subset Z$ for every vector dominates itself.
By strong duality, the optimal of the primal and the dual are equal.
Hence, by substituting the primal problem (\ref{opt_3}) with the dual problem (\ref{opt_4}),
we can rewrite the original optimization problem (\ref{opt_2}) equivalently as follows.
\begin{equation}\label{opt_5}
\max_{\lambda,z(\lambda, \beta, p)}\sum_{i=1}^{n} z_i f_i-\lambda\epsilon
\end{equation}
\vspace{-10pt}
\begin{eqnarray*}
z_i-z_{i-1}\leq \lambda (v_i - v_{i-1}), && \forall 1\leq i \leq n\\
z\in Z, 
\lambda\geq0. &&
\end{eqnarray*}

To solve problem (\ref{opt_5}), we are going to reduce it into another problem:
\begin{equation}\label{opt_66}
\max_{\lambda,z}\sum_{i=1}^{n} z_i f_i-\lambda\epsilon
\end{equation}
\vspace{-10pt}
\begin{eqnarray*}
z_i-z_{i-1}\leq \lambda (v_i - v_{i-1}), && \forall 1\leq i \leq n\\
z\in P,
\lambda\geq0. &&
\end{eqnarray*}

Intuitively, the solution to (\ref{opt_66}) is the payment rule of an IR-IC mechanism of which expected revenue minus $\lambda\epsilon$ achieves the maximum and adjacent payments are bounded by a multiple of $\lambda$. To show that the the solutions to problem (\ref{opt_5}) and (\ref{opt_66}) are the same, we need the following lemma:
\begin{lemma}\label{zinp}
If $(\lambda, z)$ is the solution of the optimization problem (\ref{opt_5}), $z\in P$.
\end{lemma}
In order to prove the above lemma, we make the following observations.

\begin{property}\label{lm1}
If $(\lambda, z)$ is the solution to the optimization problem (\ref{opt_5}), $z$ is non-decreasing.
\end{property}
\begin{proof}
Suppose $(\lambda^*, z)$ optimizes problem (\ref{opt_5}) and there exist $i$ such that $z_i>z_{i+1}$.
By constraint of problem (\ref{opt_5}), there exists $p'\in P$ and $p'$ dominate $z$.
We construct another vector $z^*$ such that $z^*_j = z_j, \forall j \neq i+1$, and $z^*_{i+1} = z_{i}$.
Since $p'$ is non-decreasing, $z^*_{i+1}=z^*_{i}=z_i\leq p'_i\leq p'_{i+1} \text{ and }\forall j\neq i+1,\ z^*_j=z_j\leq p'_j$. Hence $p'$ dominates $z^*$ as well.
Moreover,
\begin{eqnarray*}
z^*_{i+1}-z^*_{i}&=&0 < \lambda^* (v_{i+1} - v_{i}), \\
z^*_{i+2}-z^*_{i+1}&=&z_{i+2}-z_{i}<z_{i+2}-z_{i+1} \leq \lambda^* (v_{i+2} - v_{i+1}),  \\
z^*_{j+1}-z^*_{j}&=&z_{j+1}-z_j \leq \lambda^* (v_{j+1} - v_{j}),  \forall j\notin \{i,i+1\}.
\end{eqnarray*}
Thus, $(\lambda^*, z^*)$ satisfies all the constraints and $z^*\cdot f>z\cdot f$.
Hence we get the contradiction and $z$ is non-decreasing.
\end{proof}

\begin{property}\label{lm2}
Let $z$ be a non-decreasing sequence. If there exists $p\in P$ such that $p$ dominates $z$, we have $z \in P$.
\end{property}
\begin{proof}
For a non-decreasing sequence $z$, denoting $x$ as the allocation for the payment $p$ that dominates $z$,
we prove $z\in P$ by finding a feasible allocation $x'$ for $z$.
For any allocation and payment pair $(x, p)$, it is truthful if and only if
\begin{equation}\label{truthful}
\frac{p_i - p_j}{v_i} \leq x_i - x_j \leq \frac{p_i - p_j}{v_j}, \forall i > j.
\end{equation}
We set $x'_0 = x_0$. For any $i \in [n]$, we set
\begin{equation*}
x'_i = x'_{i-1} + \frac{z_i - z_{i-1}}{v_i}.
\end{equation*}
It is easy to check that  $x'$ is monotone and it satisfies the truthful condition in inequality \ref{truthful}.
We only need to check that $x_n \leq 1$ and it is individual rational.
First, we observe
\begin{eqnarray*}
x'_n &=& x'_0 + \sum_{1\leq i\leq n} \frac{z_i - z_{i-1}}{v_i} \\
&\leq& x_0 + \sum_{1\leq i\leq n} \frac{p_i - p_{i-1}}{v_i} 
\leq x_n \leq 1.
\end{eqnarray*}
The first inequality holds because $p$ dominates $z$, and the maximum is achieved only when $z_i = p_i$ for any $i\geq 1$.
By construction that $x'_0 = x_0$ and $z_0 = p_0$, $x'_0 v_0 - z_0 \geq 0$.
Using induction, for any $i\in [n]$, we have
\begin{eqnarray*}
x'_i v_i - z_i &=& (x'_{i-1} + \frac{z_i - z_{i-1}}{v_i}) v_i - z_i\\
&=& x'_{i-1} v_i - z_{i-1}\\& \geq &x'_{i-1} v_{i-1} - z_{i-1}.
\end{eqnarray*}
Thus the allocation $x'$ and payment $z$ satisfies individual rationality,
and $z \in P$.
%
%
\end{proof}

Combining Property \ref{lm1} and \ref{lm2}, Lemma \ref{zinp} holds.
Now we are ready to prove Theorem \ref{thm_1buy}.

\begin{proof}[Proof of Theorem \ref{thm_1buy}]
By Lemma \ref{zinp}, the solution $(\lambda^*, z^*)$ to problem (\ref{opt_66}) is the solution to problem (\ref{opt_5}). 
Now, we can view $z^*$ as the payment rule of a mechanism and let $p^*_i=z^*_i, \beta^*_i=0, \forall 1\leq i\leq n$. Thus, if a mechanism has the same payment rule as $z^*$, it is an optimal mechanism to problem (\ref{opt_5}). So all we have to do now is find a mechanism that optimizes problem (\ref{opt_66}), which is the solution to problem (\ref{opt_5}) as well.

Let $z\in P$ and $y$ be the corresponding allocation rule such that for any $1 \leq i \leq n$, $y_i - y_{i-1} = \frac{z_i - z_{i-1}}{v_i}$, and $y_0 = \frac{z_0}{v_0}$.
By Myerson's Lemma for discrete distributions, we have
\begin{equation}
\sum_{i=1}^n z_i f_i - \lambda\epsilon = \sum_{i=1}^n y_i \phi_i f_i - \lambda\epsilon. \label{revenue}
\end{equation}

Let $k$ be the index that $\phi_k\geq 0$ and $\forall i<k,\ \phi_i<0$.
Fixing $y_k$ and $\lambda$,
by the constraint of problem (\ref{opt_66}),
we have $y_j-y_{j-1} = \frac{z_j - z_{j-1}}{v_j} \leq \frac{\lambda(v_j - v_{j-1})}{v_j}$.
Therefore, for any $i>k, y_i\leq \min\{y_k+\sum_{j=k+1}^i\frac{\lambda(v_j - v_{j-1})}{v_j},1\}$.
In order to maximize the objective, all the above inequalities should be equalities.
Similarly, for any $i<k, y_i\geq \max\{y_k-\sum_{j=i+1}^k\frac{\lambda(v_j - v_{j-1})}{v_j},0\}$,
and the equalities should hold to maximize the objective.
When $\lambda = 0$,
we can easily get that for any $0 \leq i \leq n$,
$y_i = y_k$, which is optimized at $y_k = 1$.
This characterization is equivalent to the case 1 in Theorem \ref{thm_1buy}.

When $\lambda > 0$, there exist $0\leq a\leq b\leq n$, such that
\begin{equation*}
y_i=\begin{cases}
0& 0 \leq i < a\\
y_k-\sum_{j=i+1}^k\frac{\lambda(v_j - v_{j-1})}{v_j} & a \leq i < k\\
y_k+\sum_{j=k+1}^i\frac{\lambda(v_j - v_{j-1})}{v_j} & k \leq i \leq b\\
1& b < i \leq n
\end{cases}
\end{equation*}

Fixing $a$ and $b$, the allocation $y$ can be uniquely determined by the value $y_a$.
Substituting equation \ref{revenue} and the characterization of variable $y$ back into problem (\ref{opt_66}), we can get the following optimization problem with $y_a$ as a variable.
\begin{equation}\label{opt_8}
\max_{\lambda, y_a}
\ \ y_a f_a \phi_a + \sum_{i=a+1}^{b} (y_a + \sum_{j=a+1}^i \frac{\lambda(v_j - v_{j-1})}{v_j}) f_i \phi_i
+ \sum_{i=b+1}^{n} f_i \phi_i -\lambda\epsilon
\end{equation}
\vspace{-15pt}
\begin{eqnarray*}
y_a + \sum_{j=a+1}^b \frac{\lambda(v_j - v_{j-1})}{v_j} \leq 1, &\\
y_a + \sum_{j=a+1}^{b+1} \frac{\lambda(v_j - v_{j-1})}{v_j} \geq 1, & \mbox{if } b < n\\
y_a \leq \frac{\lambda(v_a - v_{a-1})}{v_a}, & \mbox{if } a > 0\\
y_a \geq 0, 
\lambda > 0. &
\end{eqnarray*}

Since the objective and constraints are linear in $\lambda$ and $y_a$,
the optimum is reached at the boundary of the domain.
%
%
Moreover, when $\lambda > 0$,
for the boundary case, $y_a = 0$ or $y_a = \frac{\lambda(v_a - v_{a-1})}{v_a}$.
In both cases,
we can get that there exist $0 \leq a' \leq b' \leq n$ such that $\lambda = 1 / (\sum_{i=a'}^{b'} \frac{v_i - v_{i-1}}{v_i})$,
and the allocation $y$ takes the following form.
The payment $z$ can be computed accordingly.
\begin{equation*}
y_i=\begin{cases}
0& 0 \leq i < a'\\
\sum_{i=a'}^j \frac{\lambda(v_j - v_{j-1})}{v_j} & a' \leq i \leq b'\\
1& b' < i \leq n.
\end{cases}
\end{equation*}
\begin{equation*}
\hspace{-0.23in}
z_i=\begin{cases}
0& 0 \leq i < a'\\
\lambda (v_i-v_{a'})& a' \leq i \leq b'\\
\lambda (v_{b'}-v_{a'})& b' < i \leq n.
\end{cases}
\end{equation*}
This characterization is equivalent to the case 2 in Theorem \ref{thm_1buy}.
By brute force searching all possible pair of $a, b$ and select the one with highest expected revenue,
we find the desired mechanism that maximizes the objective value of problem (\ref{opt_5}) and (\ref{opt_66}).
Therefore, we have for any $0 \leq i \leq n$, $x_i = y_i, p_i = z_i$
and the mechanism that optimizes problem (\ref{opt_5}) is also an optimal solution for problem (\ref{opt_1}).
Thus finishes the proof of Theorem \ref{thm_1buy}.
\end{proof}

\subsection{Irregular Distribution}\label{sec_irregular}

In this section, we show the result for the single buyer case even when the prior distribution known by the seller is irregular.
First we state the theorem.
\begin{theorem}\label{thm_1buy_irregular}
When there is a single buyer and $\bM$ is the set of all IR-IC mechanisms,
for any continuous distribution $f$,
mechanism
$\cM^* = \arg\max_{\cM \in \bM}\min_{f'\in EMD(f, \epsilon)} Rev(\cM, f')$
with allocation and payment $(x, p)$
takes the following form:
there exist $s \geq 1$ and a set
$\{(a_1,b_1), \dots, (a_s,b_s)\}$
such that $\lambda = 1 / (\sum_{i=1}^s \ln \frac{b_i}{a_i})$, and
\begin{eqnarray*}
x(v) &=&\begin{cases}
0 & 0 \leq v < a_1\\
\lambda (\sum_{j=1}^{i-1} \ln \frac{b_j}{a_j} + \ln \frac{v}{a_i}) \ \ \ \ \ \
& a_i \leq v \leq b_i, \forall 1 \leq i \leq s\\
\lambda \sum_{j=1}^i \ln \frac{b_j}{a_j}
& b_i < v < a_{i+1}, \forall 1 \leq i \leq s-1\\
1 & b_s < v,
\end{cases} \\
\end{eqnarray*}
\end{theorem}

Note that given the allocation rule, the payment can be computed using Myerson's payment identity \cite{myerson1981optimal}. 
Here, instead of proving that Theorem \ref{thm_1buy_irregular} is correct using duality,
we present another intuitive idea to show the correctness of the theorem.
This idea will help us have a further understanding of the result.
First, we plot the revenue curve of the known distribution,
where the $x$-axis is the quantile of the distribution,
i.e., the probability that is larger than or equal to a certain value.

\begin{figure}[htbp]
\vspace{10pt}
\begin{center}
\centering
\setlength{\unitlength}{1.3cm}
\thinlines
\begin{picture}(6.2,4)

\put(0,0){\vector(1,0){6.4}}
\put(0,0){\vector(0,1){3.9}}

\put(-0.4,2.5){$\gamma$}
\put(0,4.1){$R(q)$}
\put(1.2,3.3){$F$}
\put(1.2,2.3){$G$}

{\thicklines
\qbezier(0,0)(0.5,5)(3,2)
\qbezier(3,2)(3.5,2)(4,3.5)
\qbezier(4,3.5)(4.5,4)(6,0)
\put(0.65,2.6){\line(1,0){1.8}}
\put(3.67,2.6){\line(1,0){1.23}}
}

\put(0.59,-0.3){$q_s$}
\put(2.38,-0.3){$q'_s$}
\put(3.6,-0.3){$q_1$}
\put(4.84,-0.3){$q'_1$}

\put(0,-0.3){$0$}
\put(5.9,-0.3){$1$}
\put(6.4,-0.3){$q$}

\multiput(0,2.6)(0.2,0){25}{\line(1,0){0.1}}

\multiput(0.65,0)(0,0.2){13}{\line(0,1){0.1}}
\multiput(2.44,0)(0,0.2){13}{\line(0,1){0.1}}
\multiput(3.67,0)(0,0.2){13}{\line(0,1){0.1}}
\multiput(4.9,0)(0,0.2){13}{\line(0,1){0.1}}

\end{picture}
\end{center}
\caption{The revenue curve of the known distribution $F$ and the distribution $G$ with minimum revenue.}
\label{fig:regular:2}
\end{figure}
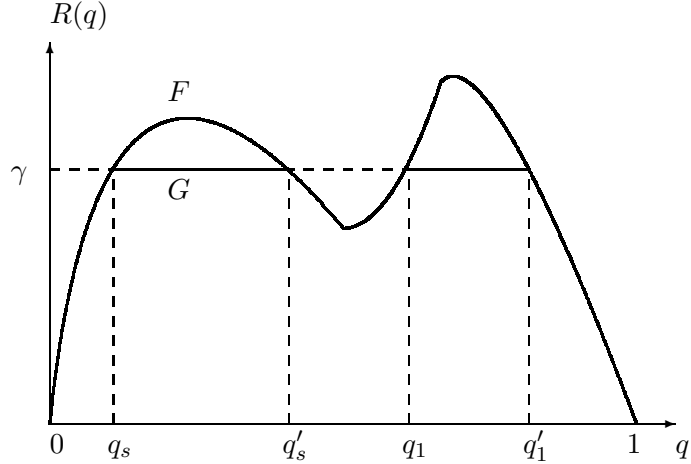

From the graph, we can easily verify that for any constant $\epsilon$,
there exists a unique $\gamma$ that intersects the
distribution at quantiles in set $Q = \{(q_1, q'_1), \dots, (q_s, q'_s)\}$,
and the resulting distribution $G$ satisfies that $EMD(G,F) = \epsilon$.
Letting $\{(a_1,b_1), \dots, (a_s,b_s)\}$ be the value that corresponds to
the set $Q$ under distribution $F$,
according to the characterization of the single item optimal Bayesian mechanism \cite{myerson1981optimal,hartline2007profit},
the mechanism $\cM^*$ designed in Theorem \ref{thm_1buy_irregular} is optimal for distribution $G$.
Moreover, according to the payment rule defined in Theorem \ref{thm_1buy_irregular},
the distribution with minimum revenue is exactly distribution $G$.
Therefore, considering any other mechanism $\cM'$, we have
\begin{eqnarray*}
\min_{f'\in EMD(f, \epsilon)} Rev(\cM', f') &\leq &Rev(\cM', G) \leq Rev(\cM^*, G) \\
&=& \min_{f'\in EMD(f, \epsilon)} Rev(\cM^*, f'),
\end{eqnarray*}
and Theorem \ref{thm_1buy_irregular} holds.

\subsection{Deterministic Mechanism}
In real world, the complex randomized algorithm may not be applicable.
Instead, people will try to approximate the optimal with simple deterministic mechanisms.
When the distribution is known and there is a single item,
the deterministic mechanism is indeed optimal.
Here in this section, we will show that when
the known distribution is within
$\epsilon$ earth mover's distance of the true distribution,
the deterministic mechanism is still a good approximation to the optimal randomized mechanism with respect to the max-min objective.
Formally, we have the following theorem.
Note that in this section, we do not require the distribution known by the seller is regular.

\begin{theorem}\label{thm_d_up}
When there is a single buyer and $\bM$ is the set of all IR-IC mechanisms,
for any distribution $f$, let
$$R = \max_{\cM\in \bM}\min_{f'\in EMD(f, \epsilon)} Rev(\cM, f')$$
For any constant $\epsilon \leq R$, there exists a deterministic mechanism $\cM$ such that
$$\min_{f'\in EMD(f, \epsilon)} Rev(\cM, f') \geq R - 2\sqrt{\epsilon R} + \epsilon.$$
\end{theorem}

\begin{proof}
It is obvious that $R \leq OPT(f)$, where $OPT(f)$ is the optimal revenue of $f$.
So we only need to prove that
$$\min_{f'\in EMD(f, \epsilon)} Rev(\cM, f') \geq OPT(f) - 2\sqrt{\epsilon \cdot OPT(f)} + \epsilon.$$

Note that for distribution $f$, there exists a simple deterministic mechanism that achieve optimal revenue.
Assume that this mechanism posts price $a$ and sells the item when the buyer's value is larger than or equal to $a$.
We assume the item is sold with probability $b$.
Clearly, $ab = OPT(f)$.
Suppose we have a new mechanism that posts price $x \leq a - \frac{\epsilon}{b}$.
Note that since $x \geq 0$, $\epsilon \leq ab$.
Next we bound the probability $\Pr_{v\sim f' \in EMD(f,\epsilon)}(v \geq x)$.
Note that the probability of $v\geq x$ is at least the probability of $v \geq a$ in distribution $f$
minus the probability that can be moved from above $a$ to below $x$.
Hence, we have
\begin{eqnarray*}
\Pr_{v\sim f' \in EMD(f,\epsilon)}(v \geq x)
\geq \Pr_{v\sim f}(v \geq a) - \frac{\epsilon}{a-x}
= b - \frac{\epsilon}{a-x}.
\end{eqnarray*}
Therefore,
$$\min_{f'\in EMD(f, \epsilon)} Rev(\cM, f')
\geq x (b - \frac{\epsilon}{a-x}). $$
By optimizing $x$ and set it as $a - \sqrt{\frac{a\epsilon}{b}}$, we have
\begin{eqnarray*}
\min_{f'\in EMD(f, \epsilon)} Rev(\cM, f')
&\geq& ab - 2\sqrt{ab\epsilon} + \epsilon \\
&=& OPT(f) - 2\sqrt{\epsilon \cdot OPT(f)} + \epsilon.
\end{eqnarray*}
Thus finishes the proof of Theorem \ref{thm_d_up}.
\end{proof}

\begin{theorem}\label{thm_d_low}
When there is a single buyer and $\bM$ is the set of all IR-IC mechanisms,
for any constant $\epsilon \leq R$,
there exists a distribution $f$ such that
for any deterministic mechanism $\cM$,
letting $R = \max_{\cM\in \bM}\min_{f'\in EMD(f, \epsilon)} Rev(\cM, f')$,
$$\min_{f'\in EMD(f, \epsilon)} Rev(\cM, f') \leq R - \frac{1}{2}\sqrt{\epsilon R}.$$
\end{theorem}
\begin{proof}
Consider the equal revenue distribution on support $[1,h]$.
Applying the result in Theorem \ref{thm_1buy},
we can calculate $R \approx 1-\frac{\epsilon}{\ln h}$.
When we set $h \to \infty$, we know that the maximin revenue approaches~$1$.

Now we consider the deterministic mechanism which posts price $x \leq h$.
We divide the analysis into 2 cases: $x < 1$ and $x \geq 1$.
When $x \geq 1$,
suppose $b$ is the value such that
exactly the probability between $x$ and $b$ is moved to below $x$.
Then, we have
\begin{eqnarray*}
\int_x^b (\frac{1}{v} - \frac{1}{b}) dv &=& \epsilon \\
-1 + \frac{x}{b} + \ln \frac{b}{x} &=& \epsilon.
\end{eqnarray*}
When $\epsilon$ is small,
the revenue of posting price $x$ is
$\frac{x}{b} \leq \frac{1}{1 + \sqrt{\epsilon} + \frac{\epsilon}{2}} \leq 1-\frac{\sqrt{\epsilon}}{2}$.

When $x < 1$,
also let $b$ be the value such that
exactly the probability between $x$ and $b$ is moved to below $x$.
Then, we have
\begin{eqnarray*}
\int_1^b (\frac{1}{v} - \frac{1}{b}) dv + (1-\frac{1}{b})(1-x) &=& \epsilon \\
\ln b + \frac{x}{b} - x &=& \epsilon.
\end{eqnarray*}
When $\epsilon$ is small, we have $x \approx b - \frac{\epsilon b}{b-1}$,
and the revenue of posting price $x$ is
$\frac{x}{b} \approx 1 - \frac{\epsilon}{b-1}$,
which is increasing with respect to $b$.
However, the maximum value of $b$ is reached when $x$ approaches 1.
Therefore, the maximum revenue is achieved in case 1 when $x\geq 1$,
and $\min_{f'\in EMD(f, \epsilon)} Rev(\cM, f') \leq 1-\frac{\sqrt{\epsilon}}{2}$.
The exact form of the inequality in the statement of Theorem \ref{thm_d_low} can be similarly derived by setting the known distribution $f$ as the equal revenue distribution on $[R, \infty)$.
\end{proof}

\section{Multiple Buyer Case}\label{sec_multi}
In this section, we consider the problem where the seller tries to sell a single item to $m$ i.i.d. buyers. 
Note that for the multi-buyer case, the max-min problem is not convex. 
Therefore, we cannot apply the duality approach as in the single buyer case 
to have a characterization for the distribution with minimum revenue. 
In this section, we will focus on a special class of mechanisms, second price mechanism with reserves, 
and as we will show later, the characterization for this set of mechanisms is not trivial. 
Moreover, in the single item environment, 
when the distributions are i.i.d. and regular, 
the second price mechanism with a fixed reserve is optimal for revenue maximization. 
Therefore, it is nature to conjecture that this simple mechanism still has approximately optimal performance in our robust setting. 
We leave the characterization of the approximation ratio of the second price mechanism with fixed reserve as an open problem. 
In this section, we will show how to efficiently compute the optimal robust reserve price for the second price auction. 
In our following analysis, we assume that the known distribution $f$ is a Lipschitz continuous distribution.
Note that for any distribution $f' \in EMD(f, \epsilon)$, 
$f'$ may not be Lipschitz continuous.

\subsection{Second Price Mechanism}
Before stating our computational result, we start with the characterization of the distribution within $\epsilon$ earth-mover's distance that generates the minimum revenue for the second price auction $\cM$.
This characterization will help us understand the minimization part of our max-min problem for the multi-buyer case,
which eventually will help us design the algorithms for finding the optimal reserve in the robust setting.

In second price auction, the expected revenue is the expected value of the second highest buyer. 
Intuitively, the distribution that generates minimum revenue is the one that differs from the known distribution 
only in the quantile corresponding to the second highest buyer. 
We formalize this idea and prove it in Theorem \ref{thm_sec}. 
%

\begin{theorem}\label{thm_sec}
Let $\cM$ be the second price mechanism,
let $g \triangleq \arg\min_{f' \in EMD(f, \epsilon)}Rev(\cM, f')$
denote the distribution with minimum revenue for mechanism $\cM$,
and let $G$ be its corresponding cumulative distribution.
Let $k$ be the smallest number such that $\frac{m-2}{m-1} \leq G(k)$, we have
\begin{enumerate}
\item $EMD(G, F) = \epsilon$.
\item There exists $l \geq k$ such that $G(v) = G(k)$ for any $k \leq v \leq l$, and $G(v) = F(v)$ for any $v > l$.
\item For any $v < k, G(v) = F(v)$.
\end{enumerate}
\end{theorem}
\begin{proof}
First note that if $EMD(G, F) < \epsilon$, by simply moving some probability from higher value to lower value, 
by revenue monotonicity, 
the expected revenue will decrease.
Therefore, we have $EMD(G, F) = \epsilon$.

Note that the probability that the second price is larger than a value $v$ is
$1$ minus the probability that all value is below value $v$
and the probability that only one value is above value~$v$.
That is,
\begin{eqnarray*}
\Pr(\text{second price}> v) &=& 1-F'^m (v)- m \cdot F'^{m-1}(v) \cdot (1-F'(v)) \\
&=& 1 + (m-1)F'^m (v) - m \cdot F'^{m-1}(v).
\end{eqnarray*}
Therefore, the expected revenue of running second price  mechanism $\cM$ with distribution $f'$ is
\begin{eqnarray}
&&Rev(\cM, f')
= \int_{0}^{\infty} \Pr(\text{second price}> v)dv \nonumber \\
&=& \int_{0}^{\infty} (1 + (m-1)F'^m (v) - m \cdot F'^{m-1}(v)) dv. 
\label{eq:rev} 
\end{eqnarray}

Since we require that distribution $f' \in EMD(f, \epsilon)$,
we know that distribution $f'$ satisfies the constraints that
\begin{equation}
F'(\infty)=1,
\int_{v=0}^{\infty} |F'(v) - F(v)|  dv \leq \epsilon,
\text{and } F'(v') \leq F'(v), \forall v' \leq v.
\label{constr}
\end{equation} 
Moreover, since moving the distribution from a lower value to a higher value will not reduce the revenue.
In order to find the distribution that generates minimum revenue, we can assume without loss of generality that
$F(v) \leq F'(v), \forall v$.
Next we show that with this assumption, 
minimizing equation \ref{eq:rev} subject to constraint \ref{constr} is equivalent to 
minimizing equation \ref{eq:rev} subject to the following constraint, 
with a changing of variable $F^*$. 
\begin{equation*}
F^*(\infty)=1,
\int_{v=0}^{\infty} |F^*(v) - F(v)|  dv \leq \epsilon,
\text{and } F(v) \leq F^*(v), \forall v.
\end{equation*}
Let $J^*(y)$ be the measure of the set $\{v| v>0 \text{ and } F^*(v) \leq y\}$
and let $F'(v) = \inf\{y | J^*(y) \geq v\}$.
Note that in constraint \ref{constr}, $F'(\infty)=1$, $F'(v)$ is monotone non-decreasing and the measure of the set
$\{v|v\geq 0\text{ and }F'(v)\leq y\}$ equals $J^*(y)$.
Moreover, let $\psi(x) = 1 + (m-1)x^m - m x^{m-1}$ and according to Lebesgue integration,
$$Rev(\cM, F') = \int_{v=0}^\infty\psi(v) \ dJ^*(v) = Rev(\cM, F^*).$$
Let $J(y)$ be the measure of the set $\{v| v>0 \text{ and } F(v) \leq y\}$.
Since for any $v$, $F^*(v) \geq F(v)$, we have
$$
\{v| v>0 \text{ and } F^*(v) \leq y\} \subseteq
\{v| v>0 \text{ and } F(v) \leq y\},
$$
which indicates that for any value $v$, $J^*(v) \leq J(v)$, and $F'(v) \leq F(v)$.
Finally, we show that
\begin{eqnarray*}
&&\int_{v=0}^{\infty} |F'(v) - F(v)| dv=\int_{v=0}^{\infty} [F'(v) - F(v)] dv \\
&=&\int_{v=0}^{\infty} [F^*(v) - F(v)] dv =\int_{v=0}^{\infty} |F^*(v) - F(v)| dv
\end{eqnarray*}
and the constructed distribution $F'$ satisfies all the constraints in (\ref{constr}).
Therefore, with the constraint $F(v)\leq F'(v)$, we can neglect the monotonicity constraint of $F'(v)$.

Now we are ready to prove the part (2) of the theorem.
Let $v'$ be the smallest number such that $\forall v \geq v',G(v) = F(v)$.
If $v' \leq k$,
by setting $l = k$,
part (1) holds.
If $v' > k$, set $l = v'$.
Assuming for contradiction that that $G(l) > G(k)$, according to the construction of $v'$,
there exists a value $s \in (k, l)$ and a sufficiently small $\epsilon'$
such that
$G(v) > F(v)$, $G(v) > G(k)$  for any $v \in [s - \epsilon', s + \epsilon']$.
We define
$$\delta=\frac{1}{2} \min_{v \in [s - \epsilon', s + \epsilon']} \min\{G(v)-F(v), G(v)-G(k)\}.$$
We note that $\delta>0$. 
Construct a new ``distribution" $G'$ with
$G'(v) = G(v) + \delta$ for any $v \in [s - \epsilon', s]$,
$G'(v) = G(v) - \delta$ for any $v \in [s, s + \epsilon']$.
We note that $G'$ may not be a real distribution since it may not be monotone. But we have
gotten rid of the monotonicity constraint.  
It is easy to verify that $G'$ satisfies the required constraint as our choice of $\delta$.
Now we define
\begin{eqnarray*}
&& D(\delta) \triangleq Rev(\cM, G) - Rev(\cM, G')\\
&=& \int_0^{\infty}
[(m-1)\cdot (G^m(v) - G'^m(v)) 
- m\cdot (G^{m-1}(v) - G'^{m-1}(v)) ] dv\\
&=& \int_{s-\epsilon'}^{s}
[(m-1)\cdot (G^m(v) - (G(v) + \delta)^{m}) 
- m\cdot (G^{m-1}(v) - (G(v) + \delta)^{m-1}) ] dv\\
&& + \int_{s}^{s+\epsilon'}
[(m-1)\cdot (G^m(v) - (G(v) - \delta)^{m}) 
- m\cdot (G^{m-1}(v) - (G(v) - \delta)^{m-1}) ] dv.
\end{eqnarray*}
Obviously, $D(0) = 0$. 
Moreover, letting $\epsilon^*$ be a constant such that $0 < \epsilon^* < \epsilon'$, when $\delta=0$, 
we have
\begin{eqnarray*}
\frac{d D(\delta)}{d\delta} \Big|_{\delta = 0}&=&
  \int_{s-\epsilon'}^{s}
  m(m-1)\cdot (G^{m-2}(v) - G^{m-1}(v)) \ dv
  +\int_{s}^{s+\epsilon'}
  m(m-1) \cdot
  (G^{m-1}(v) - G^{m-2}(v)) \ dv \\
&\geq& m(m-1) \Big[(\epsilon' - \epsilon^*)
  (G^{m-2}(s-\epsilon^*) - G^{m-1}(s-\epsilon^*)) 
  + \epsilon^*
  (G^{m-2}(s) - G^{m-1}(s)) \Big]\\
&&+ m(m-1)
  \Big[ \epsilon^*
  (G^{m-1}(s) - G^{m-2}(s)) 
  + (\epsilon' - \epsilon^*)
  (G^{m-1}(s+\epsilon^*) - G^{m-2}(s+\epsilon^*)) \Big]\\
&=& m(m-1) (\epsilon' - \epsilon^*) 
\cdot \int_{G_{s-\epsilon^*}}^{G_{s+\epsilon^*}}[(m-1)x^{m-2}-(m-2)x^{m-3}] \ dx > 0.
\end{eqnarray*}
The above inequality holds because $x^{m-2} - x^{m-1}$ 
is monotone decreasing when $x \geq \frac{m-2}{m-1}$, 
and $G(v) \geq \frac{m-2}{m-1}$ when $v \geq s - \epsilon' > k$.
Therefore, there exists a sufficiently small $\delta$ such that
$D(\delta) > 0$,
which means that the revenue of $G'$ is smaller, a contradiction.
Hence, part (2) of the theorem is correct.

The proof of part (3) is similar to the proof of part (2).
Suppose otherwise,
there exists $s$, $\iota$ and $\epsilon'$ such that for any $v\in [s-\epsilon', s+\epsilon']$, $G(v)<\frac{m-2}{m-1}-\iota$ and $G(v) > F(v) + \iota$.
Consider anther distribution $G'$
with $G'(v) = G(v)-\iota$ for any $v\in [s-\hat{\epsilon}, s]$,
$G'(v) = G(v) +\iota$ for any $v\in [s,s+\hat{\epsilon}]$.
Applying the same approach before, we can verify that
for sufficiently small $\iota$,
the revenue of $G'$ is smaller because
the second derivative of $1 + (m-1)x^m-m \cdot x^{m-1} $ is negative when $x<\frac{m-2}{m-1}$, which is a contradiction.
The detailed proof is omitted here.
Therefore $\forall j<k$, $G(j) = F(j)$.
By combining all the proofs together, Theorem \ref{thm_sec} holds.
\end{proof}

\subsection{Second Price With Fixed Reserve}
As shown above, we have characterized the distribution with minimum revenue in the second price mechanism.
By extending the result to second price mechanism with a fixed reserve,
we then show how to use this characterization as a tool to design a FPTAS algorithm for finding the optimal robust reserve.
Here we assume the value distribution has bounded support.

\begin{theorem}\label{thm_reserve}
For any constant $\epsilon, \epsilon'$, if the buyers' value distribution is Lipschitz continuous with support in $[0, H]$,
and $\bM$ is the set of second price mechanisms with reserves,
there exists a polynomial time ($poly(m, \frac{1}{\epsilon'}, H)$) algorithm for finding mechanism $\cM^*$ such that
\begin{equation*}
\min_{f'\in EMD(f, \epsilon)} Rev(\cM^*, f')
\geq \max_{\cM \in \bM}\min_{f'\in EMD(f, \epsilon)} Rev(\cM, f') - \epsilon',
\end{equation*}
where $\bM$ is the set of second price mechanisms with fixed reserves.
\end{theorem}

Note that we do not have any close form formula for the minimum revenue within $\epsilon$ earth mover's distance.
In order to find the approximately optimal reserve,
we need to limit the possible choice of reserves. 
Since the buyers' values locate in support $[0,H]$,
we consider the reserve being the multiply of $\epsilon_1$.
That is, we only consider the reserve in the set $R = \{i \cdot \epsilon_1\}_{i\in [\frac{H}{\epsilon_1}]}$.
We prove that this is sufficient to compute the approximately optimal reserve.
\begin{lemma}\label{lem_discrete}
There exists a second price mechanism $\cM'$ with reserve $r' \in R$ such that
\begin{equation*}
\min_{f'\in EMD(f, \epsilon)} Rev(\cM', f')
\geq \max_{\cM \in \bM}\min_{f'\in EMD(f, \epsilon)} Rev(\cM, f') - \epsilon_1.
\end{equation*}
\end{lemma}
\begin{proof}
\sloppy
Let the optimal mechanism be $\cM^* = \arg\max_{\cM \in \bM}\min_{f'\in EMD(f, \epsilon)} Rev(\cM, f')$.
Suppose the mechanism $\cM^*$ has reserve $r^*$.
By simply setting $r' = \max\{r\in R | r \leq r^*\}$,
we only need to prove that for any distribution $f' \in EMD(f, \epsilon)$,
\begin{equation}\label{eq_eps_approx}
Rev(\cM', f') \geq Rev(\cM^*, f') - \epsilon_1.
\end{equation}
In fact, by construction of $R$, $r' \geq r^* - \epsilon_1$.
Therefore, for any valuation profile, the payment in $\cM'$ is at least the payment in $\cM^*$ minus $\epsilon_1$.
Hence, for any distribution $f'$, the expected payment satisfies Equation \ref{eq_eps_approx}, and Lemma \ref{lem_discrete} holds.
\end{proof}

\begin{proof}[Proof of Theorem \ref{thm_reserve}]
By Lemma \ref{lem_discrete}, we can focus on computing the minimum revenue for a fixed reserve price $r$.
To begin with, we explicitly express the revenue of a distribution in terms of its cumulative probability function.
For simplicity, we assume without loss of generality that a buyer can only get the item if he bids strictly larger than the reserve.
Then, we can rewrite the expected revenue of the second price with reserve is
\begin{align*}
&Rev(\cM, f')
= r \cdot \Pr(1^{st} \text{ price} > r) + \int_{r}^{\infty} \Pr(2^{nd} \text{ price} > v)dv\\
=& r(1 - F'^m(r)) 
 + \int_{r}^{\infty} (1 + (m-1)F'^m (v) - m \cdot F'^{m-1}(v)) \ dv.
\end{align*}

First, let $g \triangleq \arg\min_{f'\in EMD(f, \epsilon)}Rev(\cM, f')$ denote the worst distribution for the second price mechanism with reserve $\cM$,
and $G$ is its corresponding cumulative distribution.
Note that it is meaningless to move the distribution below the reserve $r$.
Therefore, for any $v \leq r$, $G(v) = F(v)$.
Moreover, similar to Theorem \ref{thm_sec}, we have $EMD(G, F) = \epsilon$.
Assuming that $k$ is the smallest number such that $G(k) \geq \max\{G(r), \frac{m-2}{m-1}\}$,
we have the following characterization.

\begin{enumerate}
\item There exists $l \geq k$ such that $G(v)= G(k)$ for any $k \leq v \leq l$, and $G(v)= F(v)$ for any $v > l$.
\item For any $r \leq v < k, G_v = \max\{F(v),G(r)\}$.
\end{enumerate}

The proof of the first property is identical to the Theorem \ref{thm_sec}.
For the second property, as stated in Theorem \ref{thm_sec}, in order to have the distribution with minimum revenue,
for any $v$, we always have $G(v) \geq F(v)$.
By monotonicity of the cumulative probability function, we have $G(v) \geq G(r)$ for any $v\geq r$.
Therefore, For any $r \leq v < k, G(v) \geq \max\{F(v),G(r)\}$.
Assuming that the equality does not hold,
there exists $s$ and $\epsilon'$ such that for any $v\in [s-\epsilon', s+\epsilon']$, $G(v) > \max\{F(v),G(r)\}$.
Similar to the proof of Theorem \ref{thm_sec},
we can construct another valid distribution $G'$ with smaller revenue,
where $G'(v) = G(v)-\delta$ for any $v\in [s, s+\epsilon']$,
$G'(v) = G(v) + \delta$ for any $v\in [s - \epsilon', s]$,
and $\delta$ is also a sufficiently small constant.

With these characterization, we know that
it is sufficient to compute the value $G_r, k$ and $l$ to determine the distribution $G$.
Note that those three variable satisfies the constraint that $EMD(G, F) = \epsilon$.
Therefore, it is sufficient for us to determine $k$ and $l$.

First, we discretize the cumulative probability space $[0,1]$ into
$Q = \big\{\big[\frac{\epsilon_2\cdot i}{m^2 H}, \frac{\epsilon_2\cdot (i+1)}{m^2 H}
\big]\big\}_{i\in[\frac{m^2 H}{\epsilon_2}]}$.
We show that for any mechanism $\cM$,
there exists a distribution $G'$ with $k'$, $l'$ such that
$k' = \max\{i | F(i) \in Q, i \geq k\}$,
$l' = \max\{i | F(i) \in Q, i \leq l\}$, and
\begin{equation*}
Rev(\cM, G')
\leq \min_{f'\in EMD(f, \epsilon)} Rev(\cM, f') + O(\epsilon_2).
\end{equation*}

If the above statement is true, then by brute force searching all possible combinations of $k,l$,
we can approximately estimate the minimum revenue for mechanism $\cM$.
Note that by our construction, $G(v) - \frac{2\epsilon_2}{m^2 H} \leq G'(v) \leq G(v)$ for any $v \in [k, l]$,
$G'(r) \geq G(r)$, and
\begin{eqnarray*}
&& Rev(\cM, G') - Rev(\cM, G) \\
&\leq&
\int_{k}^{l} ((m-1)G'^m(v) - m \cdot G'^{m-1}(v) 
- ((m-1)G^m(v) - m \cdot G^{m-1}(v))) dv \\
&\leq& (l - k) ((m-1)G'^m(k) - m \cdot G'^{m-1}(k) 
- ((m-1)G^m(k) - m \cdot G^{m-1}(k))) \\
&\leq& mH (G^{m-1}(k) - G'^{m-1}(k)) \leq O(\epsilon_2).
\end{eqnarray*}
By carefully choosing $\epsilon_1$ and $\epsilon_2$,
Theorem \ref{thm_reserve} holds.
\end{proof}

\section{Conclusion and Future Direction}
We characterized the optimal robust mechanism for single-buyer case 
and the optimal robust reserve price for the second price auction for the multi-buyer case. 
However, many interesting questions are left open.
For starters, it would be exciting to bound the gap between the worst case performance of the second price mechanism with a fixed reserve and the optimal robust mechanism.
Moreover, we do not have any characterization of the optimal robust mechanism 
for this max-min goal when there are more than one buyer.
For this problem, we conjecture that when the known distribution is regular,
second price mechanism with randomized reserves is optimal.
Another interesting direction is to investigate this problem in the multi-parameter setting.
The revenue maximization problem for multi-parameter setting is notoriously hard.
So, a good starting point would be analyzing the performance of the simple constant approximation mechanisms for this case.


\bibliographystyle{apalike}  
\bibliography{sample-bibliography}  

\end{document}